\documentclass[11pt,letterpaper]{article}

\usepackage{url}
\usepackage{fullpage}

\usepackage[subtle,mathspacing=normal]{savetrees}
\usepackage{amsmath}
\usepackage{algorithm}
\usepackage[noend]{algpseudocode}
\usepackage{amssymb}
\usepackage{amsthm}
\usepackage{color}
\usepackage{array}
\usepackage{xy}
\usepackage{setspace}
\usepackage{varwidth}
\usepackage{multicol}
\usepackage{algorithm}
\usepackage{hyperref}
\usepackage{thm-restate}
\usepackage{todonotes}

\newcommand{\poly}{\operatorname{poly}}

 \newcommand{\E}{\mathbb{E}}

\newcommand{\dmax}{d_{\text{max}}}
\newcommand{\dcore}{d_{\text{core}}}
\newcommand{\ecore}{\epsilon_{\text{core}}}
\newcommand{\choice}{\texttt{choice}}

\newtheorem{thm}{Theorem}[section]

\theoremstyle{remark}

\newtheorem{theorem}{Theorem}[section]

\newtheorem{lemma}[thm]{Lemma}
\newtheorem{proposition}[thm]{Proposition}
\newtheorem{claim}[thm]{Claim}

\theoremstyle{remark}
\newtheorem{remark}[thm]{Remark}

\newcommand{\defn}[1]{\textbf{\emph{#1}}}
\renewcommand{\paragraph}[1]{\vspace{.2 cm} \noindent \textbf{#1}}

\usepackage{authblk}

\begin{document}

\title{Efficient $d$-ary Cuckoo Hashing at High Load Factors by Bubbling Up}
\author[1]{William Kuszmaul}
\author[2]{Michael Mitzenmacher}
\affil[1]{Department of Computer Science, Carnegie Mellon University \newline \texttt{wkuszmau@andrew.cmu.edu}}
\affil[2]{School of Engineering and Applied Sciences, Harvard University \newline \texttt{michaelm@eecs.harvard.edu}}
\date{}
\maketitle

\begin{abstract}
A $d$-ary cuckoo hash table is an open-addressed hash table that stores each key $x$ in one of $d$ random positions $h_1(x), h_2(x), \ldots, h_d(x)$. In the offline setting, where all items are given and keys need only be matched to locations, it is possible to support a load factor of $1 - \epsilon$ while using $d = \lceil \ln \epsilon^{-1} + o(1) \rceil$ hashes. The online setting, where keys are moved as new keys arrive sequentially, has the additional challenge of the time to insert new keys, and it has not been known whether one can use $d = O(\ln \epsilon^{-1})$ hashes to support $\poly(\epsilon^{-1})$ expected-time insertions. 

In this paper, we introduce \emph{bubble-up cuckoo hashing}, an implementation of $d$-ary cuckoo hashing that achieves all of the following properties simultaneously:
\begin{itemize}
    \item uses $d = \lceil \ln \epsilon^{-1} + \alpha \rceil$ hash locations per item for an arbitrarily small positive constant $\alpha$.
    \item achieves expected insertion time $O(\delta^{-1})$ for any insertion taking place at load factor $1 - \delta \le 1 - \epsilon$.
    \item achieves expected positive query time $O(1)$, independent of $d$ and $\epsilon$.

\end{itemize}
The first two properties give an essentially optimal value of $d$ without compromising insertion time. The third property is interesting even in the offline setting: it says that, even though \emph{negative} queries must take time $d$, \emph{positive} queries can actually be implemented in $O(1)$ expected time, even when $d$ is large.

\end{abstract}

\thispagestyle{empty}
\setcounter{page}{0}
\newpage

\section{Introduction}

In recent decades, cuckoo hashing has emerged as one of the most widely used studied hash table designs for both theory and practice (see, e.g., \cite{arbitman2010backyard,bell20241,cui2017spacf,
dietzfelbinger2007balanced,
fan2014cuckoo,fotakis2005space,
fountoulakis2013insertion,
frieze2011analysis,kirsch2010more,
li2014algorithmic, naor2008history,
pagh2004cuckoo,walzer2022insertion} as examples) . In its most basic form \cite{pagh2004cuckoo}, cuckoo hashing is a technique for using two hash functions $h_1, h_2$ in order to place some number $m$ of elements into an array of size $n$. The invariant that the hash table maintains is that, at any given moment, every element $x$ is in one of positions $h_1(x), h_2(x) \in [n]$. Of course, it is not obvious that such an invariant is even possible. What Pagh and Rodler showed in their seminal 2001 paper \cite{pagh2004cuckoo} was that, so long as $m < n/2 - \Omega(n)$, then not only is the invariant possible (with probability $1 - O(1/n)$), but it is even possible to support insertions in $O(1)$ expected time.

The advantage of cuckoo hashing is its query time: every query completes in just $2$ memory accesses. The main disadvantage -- at least for the simplest version of the data structure -- is its space efficiency: the hash table behaves well only if the load factor $m/n$ is less than $1/2$. When the load factor surpasses $1/2$, with high probability one cannot place the elements into the hash table and maintain the invariant.

To support larger load factors, one can use $d > 2$ hash functions $h_1, h_2, \ldots, h_d$. This version of the data structure, known as \defn{$d$-ary cuckoo hashing}, has been studied in both the offline (i.e., all of the elements are given to us up front) \cite{dietzfelbinger2010tight, fountoulakis2010orientability, frieze2012maximum} and the online (i.e., the elements are inserted one by one) \cite{fotakis2005space, frieze2011analysis, fountoulakis2013insertion, frieze2019insertion, bell20241, walzer2022insertion, eppstein2014wear, Khosla2019faster} settings. 
Much of the work in the offline setting has focused on establishing the minimum value of $d$ needed to support any given load factor $1 - \epsilon$. As $\epsilon \rightarrow 0$, the optimal value for $d$ becomes $\lceil \ln \epsilon^{-1} + o(1) \rceil$, where the $o(1)$ term is a function of $\epsilon^{-1}$ \cite{dietzfelbinger2010tight, fountoulakis2010orientability, frieze2012maximum}.

The main challenge in the online setting is to support a small value for $d$ while \emph{also} supporting fast insertions as a function of $\epsilon^{-1}$. Here, there are two high-level goals:
\begin{itemize}
\item \textbf{Goal 1:} Use a small value of $d$, ideally close to $\ln \epsilon^{-1}$.
\item \textbf{Goal 2:} Support insertions in expected time close to $O(\epsilon^{-1})$.
\end{itemize}

On the upper-bound side, there have been two major steps forward so far: that one can support $d = \Theta(\log \epsilon^{-1})$ while offering expected insertion time $\epsilon^{-O(\log \log \epsilon^{-1})}$ \cite{fotakis2005space}; and that one can support $d = \lceil \ln \epsilon^{-1} + o(1)\rceil$ while offering expected insertion time $f(\epsilon^{-1})$ for some unknown but potentially large function $f$ \cite{bell20241}. Whether one can achieve an expected insertion time of the form $O(\epsilon^{-1})$, or even of the form $\poly(\epsilon^{-1})$, while using $d = \ln \epsilon^{-1} + O(1)$, or even $d = O(\ln \epsilon^{-1})$, has remained open.

Finally, when it comes to \emph{queries}, there is also a third goal that one might ask for:
\begin{itemize}
\item \textbf{Goal 3: }Support \emph{positive} queries in \emph{expected} time $O(1)$, independent of $d$ and $\epsilon$.
\end{itemize}

This final goal is a bit subtle. If we query an element $x$ that is not present (this is a \defn{negative} query), then we cannot help but spend $d$ time on the query. If we query an element $x$ that is present (this is a \defn{positive} query), then the \emph{worst-case} query time is $d$, but the \emph{expected} query time could be smaller, since the query can stop as soon as it finds the element it's looking for. In principle, one could hope for an expected positive query time of $O(1)$.

It is not clear, even intuitively, whether one should expect these three goals to be compatible. One might imagine that there is tension between Goals 1 and 2, for example, that $O(\epsilon^{-1})$ insertions are possible when $d = (1 + \Omega(1)) \ln \epsilon^{-1}$ but not when $d = (1 + o(1)) \ln \epsilon^{-1}$. There could also be tension between Goals 1 and 3: even in the offline setting, it is conceivable that, to achieve $d = \ln \epsilon^{-1} + O(1)$, one must place each element $x$ in a roughly random position out of $h_1(x), \ldots, h_d(x)$. This, in turn, would force an expected positive query time of $\Omega(d)$. 

\paragraph{This Paper: bubble-up cuckoo hashing.} In this paper, we present \defn{bubble-up cuckoo hashing}, an implementation of $d$-ary cuckoo hashing that achieves all three of the above goals simultaneously:

\begin{theorem}[Restated later as Theorem \ref{thm:main}]
Let $\alpha \in (0, 1)$ be a positive constant. Let $\epsilon \in (n^{-1/4}, 1)$ be sufficiently small as a function of $\alpha$ (i.e., $\epsilon$ is at most a small constant). There exists an implementation of $d$-ary cuckoo hashing that:
\begin{itemize} 
\item uses $d = \lceil \ln \epsilon^{-1} + \alpha \rceil$; 
\item achieves expected insertion time $O(\delta^{-1})$ for any insertion taking place at a load factor $1 - \delta \le 1 - \epsilon$;
\item achieves expected positive query time $O(1)$;
\item can support $(1 - \epsilon)n$ insertions with a total failure probability $n^{-\Omega(1)}$. 
\end{itemize}
\label{thm:mainintro}
\end{theorem}

To parse Theorem \ref{thm:mainintro}'s bounds, it is helpful to consider what happens if we set $\alpha$ to be, say, $0.1$. In this case, the first two bullets of the theorem say that, for all sufficiently small $\epsilon$, it is possible to support load factor up to $1 - \epsilon$ using $d = \lceil \ln \epsilon^{-1} + 0.1\rceil$, and while supporting efficient insertions. For perspective, even \emph{offline} solutions require $d \ge \lceil \ln \epsilon^{-1} \rceil$. Thus, if $\lceil \ln \epsilon^{-1} + 0.1\rceil = \lceil \ln \epsilon^{-1} \rceil$, then Theorem \ref{thm:mainintro} achieves the \emph{exact optimal offline} $d$, while supporting efficient insertions.

The third bullet point of Theorem \ref{thm:mainintro}, on the other hand, bounds query time: it says that, for any element $x$ that is in the hash table, the expected number of probes needed for a query to find it is $O(1)$. Critically, this $O(1)$ time bound holds even if $d$ and $\epsilon^{-1}$ are large -- that is, it treats all three of $d, \epsilon^{-1}, n$ as asymptotic parameters. 

In specifying an implementation of $d$-ary cuckoo hashing, there are two algorithmic knobs that we can play with. The first algorithmic knob is the \defn{insertion policy}: when we need to place an item $x$ in one of its positions $h_1(x), \ldots, h_d(x)$, the insertion policy chooses which position to use. This choice is especially important in the (common) case where all $d$ positions are already occupied, in which case what we are really choosing is which of $d$ elements we are going to ``evict''. The evicted element will, in turn, need to pick from one of its $d$ choices, and so on. The second algorithmic knob is the \defn{query policy}: when we query an element $x$, the query policy decides in what order we should examine the positions $\{h_i(x)\}$? This may seem like a silly distinction at first glance (why not just use the order $i = 1, 2, 3, \ldots$?), but it will turn out to be surprisingly important for bounding the expected time for positive queries (the third bullet point of Theorem \ref{thm:mainintro}).

The reason that we call our algorithm \emph{bubble-up cuckoo hashing} is because of how it implements insertions. For any given element $x$, the choice $\choice(x)$ for which hash function $h_{\choice(x)}$ it uses has a tendency to ``bubble up'' over time. At any given moment, there is some value $\dmax$ dictating the maximum value of $\choice(x)$ over all elements $x$ in the table. When an element $x$ is evicted, it prefers to increase its value of $\choice(x)$. Only when an element reaches $\choice(x) = \dmax - O(1)$ does the element become willing to have $\choice(x)$ decrease; and even then, it keeps $\choice(x)$ within $O(1)$ of $\dmax$. Over time, the parameter $\dmax$ also increases, so that elements $x$ that were ``at the surface'' (i.e., $\choice(x)$ was $\dmax - O(1)$) can once again continue bubbling up.

The quantity $\dmax$ ends up also being important for our query policy: rather than examining the hashes in the order of $h_1(x), h_2(x), \ldots$, we examine them in the order of $h_{\dmax}(x), h_{\dmax - 1}(x), h_{\dmax - 2}(x), \ldots$. Intuitively, because elements $x$ tend to ``bubble up'' in their value of $\choice(x)$ over time, most elements $x$ will be in a position of the form $h_{\dmax - j}(x)$ for some relatively small $j$. In fact, for any element that is in the hash table, we will show that the query time is bounded by a geometric random variable with mean $O(1)$.

To present bubble-up cuckoo hashing as cleanly as possible, we will begin in Section \ref{sec:basic} with a warmup version of the algorithm which we call \defn{basic bubble-up cuckoo hashing}. Even this basic version of the algorithm achieves a nontrivial guarantee: supporting $O(\delta^{-1})$-time insertions with $d = O(\ln \epsilon^{-1})$. The advantage of starting with basic bubble-up cuckoo hashing is that it is \emph{remarkably simple}. Both the algorithm and its analysis could reasonably be taught in a data structures course.

After presenting the basic version of the algorithm, we continue to Section \ref{sec:advanced} to present \defn{advanced bubble-up cuckoo hashing}. This is the algorithm that achieves the full set of guarantees in Theorem \ref{thm:mainintro}. The analysis of the algorithm ends up requiring quite a few more ideas than its simpler counterpart, but nonetheless is able to build on some of the same basic principles.


\section{Related Work}


The idea of using $d > 2$ hash functions was first proposed by Fotakis, Pagh, Sanders, and Spirakis \cite{fotakis2005space} in 2005. The authors showed that, using the breadth-first-search insertion policy, it is possible to set $d = \Theta(\log \epsilon^{-1})$ while supporting insertions in expected time $\epsilon^{-O(\log \log \epsilon^{-1})}$.


Even in the offline setting, it is interesting to study the \defn{maximum load threshold} $c_d^*$ that a $d$-ary cuckoo hash table can support. This threshold  $c_d^*$ is the maximum value such that, as $n$ goes to infinity, it becomes possible to support any load factor of the form $(1 - \Omega(1)) c_d^*$. Several independent works \cite{dietzfelbinger2010tight, fountoulakis2010orientability, frieze2012maximum} have characterized this load threshold, showing, in particular, that is equivalent to the previously known thresholds for the random $k$-XORSAT problem \cite{dietzfelbinger2010tight}. One consequence is that, as $d \rightarrow \infty$, the threshold $c_d^*$ behaves as $1 - e^{-d} - e^{-o(d)}$. 

In the online setting, there has been a great deal of interest in analyzing the \emph{random-walk} insertion policy. Frieze, Melsted, and Mitzenmacher \cite{frieze2011analysis} study the \emph{worst-case} insertion time of this strategy. They show that, for any load factor $1 -\epsilon \in (0, 1)$, there exists $d = \Theta(\log \epsilon^{-1})$ such that the worst-case insertion time is polylogarithmic with probability $1- o(1)$. This result was subsequently tightened by Fountoulakis, Panagiotou, and Steger \cite{fountoulakis2013insertion} to support load factors closer to $c_d^*$. 

The expected insertion time of random-walk insertions has proven quite tricky to analyze. Frieze and Johansson \cite{frieze2019insertion} show that, for any constant $\epsilon \in (0, 1)$, there exists a $d_\epsilon$ such that for all $d > d_{\epsilon}$, such that random-walk insertions at a load factor of $1 - \epsilon$ support $O(1)$ expected insertion time. Very recently, Bell and Frieze \cite{bell20241} proved a stronger result: they show that, so long as $4 \le d \le O(1)$, and so long as the load factor is at most $(1 - \Omega(1)) c_d^*$, then the expected insertion time is $O(1)$. The state of the art for $d = 3$ remains a result by Walzer \cite{walzer2022insertion} who proves $O(1)$ expected insertion time for load factors up to $0.818$. 

Eppstein, Goodrich, Mitzenmacher, and Pszona \cite{eppstein2014wear} consider the task of minimizing the \emph{wear} of the hash table, defined to be the \emph{maximum number of times that any single item gets moved}. They show that, for $d = 3$, there exists an insertion algorithm that fills the hash table to load factor $\Omega(1)$ while guaranteeing a maximum wear of at most $\log \log n + O(1)$ with high probability.

Khosla and Anand \cite{Khosla2019faster} consider the task of proving a \emph{high probability} bound on the \emph{total insertion time} needed to fill a $d$-ary cuckoo hash table to load factor $1 - \epsilon$. They show that, if $d = O(1)$ and if $1 - \epsilon = (1 - \Omega(1)) c_d^*$, then there is an insertion algorithm that achieves a high-probability insertion-time bound of $O(n)$.


An alternative to $d$-ary cuckoo hashing is \emph{bucketized cuckoo hashing}, introduced by Dietzfelbinger and Weidling \cite{dietzfelbinger2007balanced} in 2007. In this setting, each key hashes to two \emph{buckets} of size $b > 1$. Dietzfelbinger and Weidling \cite{dietzfelbinger2007balanced} studied the breadth-first-search insertion policy, and showed that, for $b \ge 16 \ln \epsilon^{-1}$, the policy achieves $O(\epsilon^{-\log \log \epsilon^{-1}})$ expected-time insertions. They left as an open question whether one could prove a similar result for random-walk insertions. The closest such result is due to Frieze and Petti \cite{frieze2018balanced}, who prove the following: if $b \ge \Omega(\epsilon^{-2} \log \epsilon^{-1})$, and if insertions are implemented using random-walk evictions, then the hash table can be filled to a load factor up to $1 - \epsilon$ while supporting insertions with eviction chains of expected length $O(1)$. Bucketized cuckoo hashing has also been studied in the offline setting, where the goal is to determine the critical load threshold for any given number $d$ of buckets and any given bucket size $b$; here, results are known both for non-overlapping \cite{fernholz2007k, cain2007random, fountoulakis2016multiple, lelarge2012new} and overlapping \cite{lehman20093, walzer2023load} buckets. 


Finally, another generalization of $d$-ary cuckoo hashing is the setting in which the $d$ hashes need not be independent. Here, a common technique is the use of a \emph{backyard}: items hash to bins of some size $d_1$, and if the bin is overloaded, then the item is placed into a (much smaller) secondary data structure known as the backyard \cite{arbitman2010backyard, goodrich2012cache, bender2018bloom, peterson1957addressing, bender2023iceberg}. If the backyard itself is a cuckoo hash table (or a deamortized cuckoo hash table \cite{arbitman2009amortized}), then the resulting data structure is known as a \emph{backyard cuckoo hash table} \cite{arbitman2010backyard}. Backyard cuckoo hashing can be used to support a load factor of form $1 - \epsilon$ with $d = \tilde{O}(\epsilon^{-2})$ \cite{arbitman2010backyard, bender2023iceberg}. Although this value of $d$ is exponentially larger than the target value of $d = \log \epsilon^{-1} + O(1)$ in the current paper, backyard cuckoo hashing turns out to be nonetheless a quite useful algorithmic primitive in the design of constant-time succinct hash tables \cite{arbitman2010backyard}.

\section{Preliminaries}

\paragraph{General notation.} We say that an event occurs with \defn{high probability in $n$}, or equivalently that it occurs with probability $1 - 1 / \poly(n)$, if, for all positive constants $c$, the event occurs with probability at least $1 - O(1 / n^c)$. We will use $[a, b]$ to denote $\{a, a + 1, \ldots, b\}$ and $(a, b]$ to denote $\{a + 1, a + 2, \ldots, b\}$. 

We will often have multiple asymptotic variables (namely, $n, d, \epsilon^{-1}, \delta^{-1}$). Our convention when using asymptotic notation will be to require that all of the variables that a function depends on go to infinity, or, more specifically, that the \emph{minimum} of the variables goes to infinity. In all of our uses of asymptotic notation, the variables will have a clear relationship: $\ln \delta^{-1} \le \ln \epsilon^{-1} \le d \le n$. So, when interpreting asymptotic notation with multiple variables, it suffices to think of just $\delta^{-1}$ as going to infinity (or, if $\delta^{-1}$ is not in use, then $\epsilon^{-1}$).

\paragraph{Cuckoo hashing terminology. } A $d$-ary cuckoo hash table stores elements in an array of size $n$, where each element $x$ is guaranteed to be in one of positions $h_1(x), h_2(x), \ldots, h_d(x)$. As in past work \cite{dietzfelbinger2010tight, fountoulakis2010orientability, frieze2012maximum, fotakis2005space, frieze2011analysis, fountoulakis2013insertion, frieze2019insertion, bell20241, walzer2022insertion, eppstein2014wear, Khosla2019faster}, will assume that the hash functions $h_1, h_2, \ldots, h_d$ are fully independent and each one is fully random. 

The \defn{load factor} of a hash table is the fraction $m/n$, where $m$ is the current number of elements. We will typically use $1 - \epsilon$ to denote the maximum load factor that the hash table supports, and $1 - \delta$ to denote the current load factor. Our goal will be to support insertions in time $O(\delta^{-1})$.

When an element $x$ is inserted, the insertion will perform an \defn{eviction chain} in which $x$ is placed in some position $j_1$, the element $x_2$ formerly in position $j_1$ is placed in some position $j_2$, and so on, with the final element in the chain being moved into a free slot. Each of the elements $x_2, x_3, \ldots$ are said to have been \defn{evicted} during the insertion. Insertions are permitted to declare \defn{failure}. In practice, one can rebuild the hash table whenever a failure occurs---so long as the overall failure probability (for all insertions) is $o(1)$, the cost of these rebuilds is a low-order term in the overall cost of the insertions.\footnote{Moreover, with some care, rebuilds can be performed (essentially) in-place. See Remark \ref{rem:rebuilds}.}

For a given element $x$ and index $i \in [d]$, we say that we have \defn{first-time probed} $h_i(x)$ if, during some insertion, we have evaluated $h_i(x)$ and checked whether that slot was empty. 

For a given element $x$ in the hash table, let $\choice(x)$ denote the index $i$ of the hash function $h_i$ that $x$ is currently using. If $x$ is not in the hash table, $\choice(x) := 0$. To simplify our discussion, we will assume throughout the body of the paper that $\choice$ can be evaluated in constant time. We will then show in Appendix \ref{app:choice} how to replace $\choice$ with an explicit protocol $\choice'$ that (with a bit of additional algorithmic case handling) preserves both the time and correctness bounds established in the body of the paper.

\paragraph{Background machinery.} Our warmup algorithm in Section \ref{sec:basic} will make use of the classic $2$-ary cuckoo hashing analysis from \cite{pagh2004cuckoo}. Note that, in $2$-ary cuckoo hashing, there is only one possible eviction chain that an insertion can follow, so there is no algorithmic flexibility in the insertion strategy.

\begin{theorem}[$2$-ary cuckoo hashing \cite{pagh2004cuckoo}]
Consider a $2$-ary cuckoo hash table where we declare failure on an insertion if it takes time $\omega(\log n)$. Then, for any sequence of $n/2 - \Omega(n)$ insertions, we have that:
\begin{itemize}
    \item The probability of any insertion failing is $O(1/n)$.
    \item The expected time per insertion is $O(1)$. 
\end{itemize}
\label{thm:2ary}
\end{theorem}

We remark that the use of $\omega(\log n)$ in Theorem \ref{thm:2ary} is a slight abuse of notation: What we mean by ``declaring failure if an insertion takes time $\omega(\log n)$'' is that the insertion algorithm should select some sufficiently large positive constant $c$ and declare failure after time $c \log n$. The use of $\omega$ notation here, and throughout the paper when specifying failure conditions, is just so that we do not have to carry around extra constants needlessly.

The advanced version of our algorithm will require the use of slightly more heavy machinery. In particular, we will apply a recent result by Bell and Frieze \cite{bell20241} who analyze the $d$-ary \defn{random-walk} eviction strategy in which, when an element $x$ is evicted, it selects a random $i \in [d]$ and goes to position $h_i(x)$. The authors show that, at load factors of the form $1 - \Omega(1)$, random-walk insertions take $O(1)$ expected time: 
\begin{theorem}[Random-walk $d$-ary cuckoo hashing \cite{bell20241}]
There exists functions $D : \mathbb{N} \rightarrow \mathbb{R}^+$, $T: \mathbb{N} \rightarrow \mathbb{R}^+$, $N:\mathbb{N}\rightarrow \mathbb{R}^+$, and $F:\mathbb{N}\rightarrow \mathbb{R}^+$ satisfying $\lim_{d \rightarrow \infty} D(d) = 0$, and such that the following is true. Let $d \in \mathbb{N}$, and consider random-walk $d$-ary cuckoo hashing on $n > N(d)$ slots, where we declare failure on an insertion if it takes time $T(d) \log^{\omega(1)} n$. Consider any sequence of $n \cdot (1 - e^{-d + D(d)})$ insertions. Then:
\begin{itemize}
    \item The probability of any insertion failing is $O(n^{-F(d)})$.
    \item The expected time per insertion is $O(T(d))$. 
\end{itemize}
\label{thm:randomwalk}
\end{theorem}

As before, we are abusing notation in our description of the termination condition: terminating after $T(d) \log^{\omega(1)} n$ steps just means that we terminate after $T(d) \log^{c} n$ steps for some sufficiently large positive constant $c$. It is also worth remarking that we will be applying Theorem \ref{thm:randomwalk} only to the case of $d = O(1)$ (specifically, when applying Theorem \ref{thm:randomwalk}, $d$ will be the parameter $\dcore$ used in our algorithm, which is set to be a constant), so the failure probability $O(n^{-F(d)})$ will be $n^{-\Omega(1)}$, and the time $O(T(d))$ will be $O(1)$. 

Finally, we will also need some machinery for proving concentration bounds. Specifically, we will make use of McDiarmid's inequality \cite{mcdiarmid1989method}, which can also be viewed as Azuma's inequality applied to a Doob martingale.

\begin{theorem}[McDiarmid's Inequality \cite{mcdiarmid1989method}]
Call a function $f(x_1, x_2, \ldots, x_n): U^n \rightarrow \mathbb{R}$ $C$-Lipschitz if changing the value of a single $x_i$ can only ever change $f$ by at most $C$. Given a $C$-Lipschitz function $f$, and given independent random variables $X_1, X_2, \ldots, X_n \in U$, the random variable $F = f(X_1, \ldots, X_n)$ satisfies
$$\Pr[|F - \E[F]| \ge j C \sqrt{n}] \le e^{-\Omega(j^2)}$$
for all $j > 0$.
\label{thm:mcdiarmid}
\end{theorem}

\paragraph{A remark on deletions.} Although our focus is on insertions and queries, it is worth noting that one can also add deletions in a black-box fashion. Indeed, using \emph{tombstones} to implement deletions results in the following corollary of Theorem \ref{thm:mainintro}, proven in Appendix \ref{app:tombstones}.

\begin{restatable}{corollary}{cormain}
Let $\alpha \in (0, 1)$ be a positive constant, and let $\dcore \in O(1)$ be sufficiently large as a function of $\alpha$. Then, for any $\epsilon \le e^{-\dcore}$ satisfying $\epsilon^{-1} \le n^{1/4}$, there exists an implementation of $d$-ary cuckoo hashing that supports both insertions and deletions and that:
\begin{itemize} 
\item uses $d = \lceil \ln \epsilon^{-1} + \alpha \rceil$; 
\item achieves amortized expected time $O(\delta^{-1} \log \delta^{-1})$ for any insertion/deletion taking place at a load factor $1 - \delta \le 1 - \epsilon$;
\item achieves expected positive query time $O(1)$.
\end{itemize}
\label{cor:maindeletions}
\end{restatable}

\section{Warmup: The Basic Bubble-Up Algorithm}\label{sec:basic}

We begin by presenting a basic version of the bubble-up algorithm that achieves the following simple guarantee:

\begin{theorem}
Let $\epsilon \in (n^{-1/4}, 1)$, and let $d = \lceil 3 \ln \epsilon^{-1} \rceil + 1$. The basic bubble-up algorithm is a $d$-ary eviction policy that supports $(1 - \epsilon)n$ insertions with probability $1 - O(1/n)$, and where the expected time for the $(1- \delta)n$-th insertion is $O(\delta^{-1} + \log \epsilon^{-1})$, for any $\delta \ge \epsilon$.
\label{thm:basic}
\end{theorem}

\paragraph{The basic bubble-up algorithm. } 
The basic bubble-up algorithm uses the following strategy for inserting/evicting an element $x$:
\begin{itemize}
    \item If $\choice(x) = d$, move $x$ to $h_{d - 1}(x)$, and evict anyone who is there. \textbf{(Type 1 Move)}
    \item If $\choice(x) = d - 1$, move $x$ to $h_{d}(x)$, and evict anyone who is there. \textbf{(Type 2 Move)}
    \item If $\choice(x) < d - 1$, sequentially check if any of $h_{\choice(x) + 1}(x), \ldots, h_{d- 2}(x)$ are free slots. If any of them are free slots, use the first one found \textbf{(Type 3 Move)}; otherwise, move $x$ to $h_{d - 1}(x)$ and evict anyone who is there \textbf{(Type 4 Move)}. 
    \item Finally, if we make $\omega(\log n)$ Type 1 or Type 2 moves in a row, without finding a free slot, declare \textbf{failure}.
\end{itemize}
If a move evicts another key, that key then proceeds through the same process above for performing its own eviction.

\paragraph{Analysis. } 
At a high-level, one should think of there as being two types of elements: those using their first $d - 2$ hashes (\defn{non-core elements}); and those using their final $2$ hashes (\defn{core elements}). The non-core elements interact passively with the table: a non-core element would like to use one of its first $d - 2$ hashes, if it can, but is not willing to evict another element to do so. Once an element ``bubbles up'' to become a core element, it interacts more actively with the table: whenever a core element $x$ is evicted, it goes to whichever position $h_d(x), h_{d - 1}(x)$ it is not currently in, and it evicts any element that is there. Morally, one should think of the core elements as forming a 2-ary cuckoo hash table (the \defn{core hash table}) that \emph{lives in the same slots} as the full $d$-ary cuckoo hash table. The elements in the core hash table take priority over those that are not, meaning that core elements can evict non-core elements, but not vice-versa. 

The main step of the analysis will be bounding the number of elements that get placed into the core hash table. What we will see is that most elements are able to make use of their first $d - 2$ hashes, and that only a small fraction make it into the core hash table. This, in turn, is what allows the core hash table to function correctly, despite the fact that it is only a 2-ary cuckoo hash table. 

The analysis, and indeed all of the analyses in this paper, will make critical use of the following coupon-collector identity, proven in Appendix \ref{sec:appendix}. 

\begin{restatable}{proposition}{propcoupons}
Let $\epsilon \in (n^{-1/4}, 1)$. Suppose we sample iid uniformly random coupons $u_1, u_2, \ldots \in [n]$, stopping once we have sampled a total of $(1 - \epsilon) n$ distinct coupons. With probability $1 - 1 / \poly(n)$, the number of sampled coupons is
    $$n \ln \epsilon^{-1} \pm \tilde{O}(n^{3/4}).$$
    \label{prop:coupons}
\end{restatable}

    

We can apply Proposition \ref{prop:coupons} to our setting as follows. Each first-time probe can be viewed as a coupon, sampling a uniformly random slot in $[n]$. The $(1 - \epsilon) n$ insertions terminate once we have sampled $(1 - \epsilon)n$ distinct slots. Thus, Proposition \ref{prop:coupons} tells us that, with probability $1 - 1 / \poly(n)$, either the eviction policy fails, or the total number of first-time probes made is $n \ln \epsilon^{-1} \pm \tilde{O}(n^{-1/4}).$

    


Recall that we call an element $x$ a \defn{core} element if $\choice(x) \in \{d - 1, d\}$. We can use Proposition \ref{prop:coupons} to derive a high-probability bound on the number of core elements:

\begin{lemma}
    With probability $1 - 1 / \poly(n)$, the number of core elements is at most $(1 + o(1)) n / 3$. 
    \label{lem:corenumber}
\end{lemma}
\begin{proof}
    Let $P$ be the total number of first-time probes that we perform, and let $K$ be the total number of core elements after all the insertions are complete.\footnote{So that $K$ is well defined, if an insertion fails, consider the insertions to, at that point, all be complete.}  If an element $x$ is a core element, then we must have first-time probed all of $h_1(x), \ldots, h_{d - 1}(x)$. It follows that
    $$P \ge (d - 1)K \ge 3 K \ln \epsilon^{-1}.$$
    On the other hand, by Proposition \ref{prop:coupons}, we have with probability $1 - 1 / \poly(n)$ that 
    $$P \le (1 + o(1)) n \ln \epsilon^{-1}.$$
    Chaining together these identities, we get that 
    $$3K \ln \epsilon^{-1} \le (1 + o(1)) n \ln \epsilon^{-1}.$$
    which implies that $K \le (1 + o(1))n / 3$, as desired.
\end{proof}

An important feature of core elements is that, whether or not an element $x$ is core has nothing to do with $h_{d - 1}(x)$ or $h_d(x)$. We can formalize this in the following lemma, which we refer to as the \emph{core independence property}.

\begin{lemma}[The Core Independence Property]
    Let $x$ be an element, and let $C$ be the event that $x$ is core. The event $C$ is independent of the pair $(h_{d - 1}(x), h_d(x))$.
    \label{lem:corerandom}
\end{lemma}
\begin{proof}
This follows from two observations: (1) Once an element $x$ becomes a core element, it stays a core element; and (2) Prior to an element $x$ becoming a core element, we never evaluate $h_{d - 1}(x)$ or $h_d(x)$. 
\end{proof}

As noted earlier, we can think of the core elements as forming a \defn{core hash table} in which the only hash functions are $h_{d - 1}$ and $h_d$. An element $x$ is inserted into the core hash table when it becomes a core element. The insertion triggers a sequence of Type 1 and Type 2 moves, that correspond to an eviction chain in the core hash table. The eviction chain ends once it encounters a slot that does not contains a core element (i.e., a slot that the core hash table thinks of as empty). Thus, the core hash table operates exactly like a standard $2$-ary cuckoo hash table. 

Lemma \ref{lem:corenumber} tells us that, with high probability, the core hash table never has more than $n /3 - o(n) \le n/2 - \Omega(n)$ elements. Lemma \ref{lem:corerandom}, on the other hand, tells us that the hash functions $h_{d - 1}$ and $h_d$ used within the core hash table are fully random. Combined, the lemmas allow us to apply the classical analysis of 2-ary cuckoo hashing (Theorem \ref{thm:2ary}) to deduce that:
\begin{itemize}
    \item \textbf{Fact 1: } Each eviction chain in the core hash table has expected length $O(1)$;
    \item \textbf{Fact 2: } The probability that any eviction chain in the core hash table ever has length $\omega(\log n)$ (this includes the failure event where the insertion fails) is at most $O(1 / n)$. 
\end{itemize}

Recall that the basic bubble-up algorithm fails if there are ever $\omega(\log n)$ Type 1 and Type 2 steps in a row. Fact 2 tells us that the probability of such a failure ever occurring is $O(1/n)$. 

Fact 1, on the other hand, can be used to bound the expected insertion time, overall, within the (full) hash table. Consider the $((1 - \delta)n + 1)$-th insertion. Let $Q$ be the number of first-time probes made by the insertion. Since each first-time probe has at least a $\delta$ probability of finding a free slot, we have that
$$\E[Q] = O(\delta^{-1}).$$
We can bound the total time $T$ spent on the insertion by the sum of two terms:
\begin{itemize}
    \item $T_1$ is the time spent on Type 1 and Type 2 moves;
    \item $T_2$ is the number of first-time probes made by Type 3 and Type 4 moves.

\end{itemize}

By construction $T_2 \le O(Q)$, so $\E[T_2] = O(\delta^{-1})$. To bound $T_1$, define $J$ to be the total number of core-table eviction chains that occur during the current insertion. By Fact 1, we have that
$$\E[T_1] \le O(J).$$
Note also that each core-table eviction chain is triggered by an element $x$ becoming a core element, at which point $h_{d - 1}(x)$ experiences a first-time probe. It follows that $J \le Q$, which implies that $\E[T_1] \le O(\delta^{-1})$. Thus $\E[T_1 + T_2] = O(\delta^{-1})$, as desired.

\section{The Advanced Bubble-Up Algorithm}\label{sec:advanced}



In this section, we present the advanced bubble-up algorithm, which achieves the main result of the paper:


\begin{restatable}{theorem}{thmmain}
Let $\alpha \in (0, 1)$ be a positive constant, and let $\dcore \in O(1)$ be sufficiently large as a function of $\alpha$. Then, for any $\epsilon$ satisfying $n^{-1/4} \le \epsilon \le e^{-\dcore}$, there exists an implementation of $d$-ary cuckoo hashing that:
\begin{itemize} 
\item uses $d = \lceil \ln \epsilon^{-1} + \alpha \rceil$; 
\item achieves expected insertion time $O(\delta^{-1})$ for any insertion taking place at a load factor $1 - \delta \le 1 - \epsilon$;
\item achieves expected positive query time $O(1)$;
\item can support $(1 - \epsilon)n$ insertions with a total failure probability $n^{-\Omega(1)}$. 
\end{itemize}
\label{thm:main}
\end{restatable}

Note that the statement of Theorem \ref{thm:main} differs slightly from the statement (Theorem \ref{thm:mainintro}) in that it introduces an extra parameter $\dcore$ to relate $\alpha$ to $\epsilon$. Although the two theorems are equivalent in their meanings (all $\dcore$ does is force $\epsilon$ to be small as a function of $\epsilon$), the introduction of $\dcore$ will make the proof of Theorem \ref{thm:main} (and, specifically, the statement of our algorithm) somewhat cleaner. 



\subsection{Technical Overview: Motivating the Advanced Bubble-Up Algorithm.}

To motivate the proof of Theorem \ref{thm:main}, let us first revisit the basic bubble-up algorithm from Section \ref{sec:basic}, which used $d = 3 \ln \epsilon^{-1} + O(1)$. There are two bottlenecks preventing us from reducing $d$. 

\paragraph{First bottleneck: the core table.} The first bottleneck is that, since the core hash table is $2$-ary, it requires a load factor smaller than $1/2$. This bottleneck is easy to rectify: Just implement the core hash table as a $\dcore$-ary hash table for some large constant $\dcore$. This allows us to store up to, say, $0.99 n$ elements in the core hash table, while still ensuring that it operates efficiently \cite{bell20241}. With this modification to the basic bubble-up algorithm, one can improve the bound on $d$ to, say, $d = 1.02 \ln \epsilon^{-1} + O(1)$. Of course, although this is an improvement over $3 \ln \epsilon^{-1} + O(1)$, it is still a far cry from the bound of $\lceil \ln \epsilon^{-1} + \alpha \rceil$ that we will ultimately want -- this is where the second bottleneck comes into play.

\paragraph{Second bottleneck: the number of first-time probes by elements not in the core table.} The second bottleneck is more significant. In the analysis of the basic bubble-up algorithm, we can argue that each element $x$ in the core hash table has made at least $d - \dcore + 1$ first-time probes. But we cannot say anything about the elements $x$ \emph{not} in the core hash table. This is not just a problem with the analysis -- the elements not in the core hash table really are likely to make much fewer than $d$ first-time probes.

To see why this is a problem, recall that, overall, we need to achieve roughly $n \ln \epsilon^{-1}$ first time probes (Proposition \ref{prop:coupons}). Now imagine for a moment that both of the following facts are true: (1) we are using $d$ only slightly larger than $\ln \epsilon^{-1}$; and (2) the elements \emph{not} in the core hash table make, on average, much fewer than $d$ total probes. Then, the only way for the total number of first-time probes to be roughly $n \ln \epsilon^{-1}$ is if the vast majority of the elements are in the core hash table! This would be a disaster for the core hash table, as it is only capable of supporting load factors of the form $1 - \Omega(1)$. 

Thus, if we want to get away with a small value of $d$, in particular, a value of the form $\ln \epsilon^{-1} + O(1)$, we need to ensure that almost all of the elements, \emph{including those not in the core table}, end up achieving \emph{very close} to $d$ probes. 

To achieve this guarantee, we introduce a critical modification to the algorithm: We set a parameter $\dmax$ that increases over time, and where, at any given moment, all of the elements use only their first $\dmax$ hashes. At a high level, the rule for increasing $\dmax$ is that, whenever the load factor reaches $1 - e^{-\dmax + \alpha}$, for the current value of $\dmax$, we increase the value of $\dmax$ by $\dcore$. We refer to the maximal time window during which $\dmax$ takes a given value as a \defn{phase}. 

At any given moment, we define the ``core hash table'' to consist of the elements that are using hash functions $\{h_{\dmax - \dcore + 1}, \ldots, \dmax\}$. This means that, whenever a new phase begins, the core hash table \emph{becomes empty}. In order for an element to get added to the core hash table during a given phase, it must be part of an eviction chain that occurs \emph{during that phase}. 

Intuitively, this results in the following situation. During a given phase $i$, a large fraction of elements, say a $0.99$ fraction, make it into the core table. These elements will have probed almost all of their hashes. (In fact, by applying Proposition \ref{prop:coupons} to the core hash table, we will be able to argue that the core elements have probed not just their first $\dmax - \dcore$ hashes, but also, on average, most of their final $\dcore$ hashes.) However, even among the elements that \emph{do not} make it into the core table in phase $i$, most of them will have made it into the core table in phase $i - 1$. These elements will \emph{also} have probed all but $O(1)$ of their hashes. Among the elements that do not make it into the core table in phases $i - 1$ or $i$, most will have made it into the core table in phase $i - 2$, and so on. The result is that most elements will have probed almost all of their hashes. This is how we are able to obtain the first-time probes required by Proposition \ref{prop:coupons} without overflowing the core table, and while using a value of $d$ of the form $\lceil \ln \epsilon^{-1} + \alpha \rceil$. 

\paragraph{Analyzing the core table. }
Of course, the above intuition doesn't actually answer the question of \emph{how} we bound the load of the core hash table. This is the main technical contribution of the proof. The basic idea is to examine four quantities:
\begin{itemize}
    \item $P_1$: The total number of first-time probes made prior to phase $i$.
    \item $P_2$: The total number of first-time probes made by the core table during phase $i$.
    \item $P_3$: The total number of first-time probes made outside of the core table (i.e., to $h_i(x)$ for some $i < \dmax - \dcore$) during phase $i$.
    \item $P_4$: The total number of first-time probes made by the end of phase $i$.
\end{itemize}
The quantities $P_1$ and $P_4$ are controlled by Proposition \ref{prop:coupons}, and together force $P_2 + P_3 = P_4 - P_1$ to be roughly $n \dmax$. The second quantity $P_2$ is also controlled by Proposition \ref{prop:coupons}, and is roughly $n \ln \rho^{-1}$, where $1 - \rho$ is the final load factor of the core hash table. This gives us a relationship $n \ln \rho^{-1} + P_3 = n \dmax \pm o(n)$ that we can use to relate $\rho^{-1}$ to $P_3$. To prove an upper bound $\rho^{-1}$, it suffices to actually prove a lower bound on $P_3$. 

We will argue that $P_3$ is controlled by a relatively simple random process that can be analyzed with the help of McDiarmid's Inequality \cite{mcdiarmid1989method}. So long as at least $\Omega(n)$ elements make it into the core hash table during the current phase, we will get a high-probability $\Omega(n)$ lower bound on $P_3$. By balancing various constants appropriately, this will allow us to obtain Theorem \ref{thm:main}.

\paragraph{What about queries? }Finally, the structure of the advanced bubble-up algorithm also comes with a second advantage: positive queries can be completed in \emph{expected} time $O(1)$, no matter how large the parameter $1 - \epsilon$ may be. This is because, for a given element $x$, we have with probability at least, say, $0.9$, that $x$ uses one of hashes $h_d, \ldots, h_{d - \dcore + 1}$; and if $x$ doesn't use any of those, then with probability at least $0.9$, it uses one of hashes $h_{d - \dcore}, \ldots, h_{d - 2\dcore + 1}$, and so on. The result is that, if we examine positions $h_{\dmax}(x), h_{\dmax - 1}(x), \ldots$ one after another, then the time to find $x$ is bounded above by a geometric random variable with mean $O(\dcore) = O(1)$.

\subsection{The Advanced Bubble-Up Algorithm}

In addition to the parameter $\alpha = \Theta(1)$, which is defined in Theorem \ref{thm:main}, the advanced bubble-up algorithm will make use of two parameters $\dcore = \Theta(1)$ and $\ecore = \Theta(1)$ that are determined by $\alpha$. We will first describe the algorithm using these parameters, and then describe how to select the parameters.

Let $\gamma = \lceil \ln \epsilon^{-1} + \alpha \rceil \pmod \dcore$. We proceed in phases, starting with phase $1$, as follows. During a given phase $q$, define $\dmax := \gamma + q \dcore$ (implicitly, $\dmax$ is a function of $q$). The phase ends when we reach load factor $1 - e^{-\dmax + \alpha}$. The final phase is the one where $\dmax = \lceil \ln \epsilon^{-1} + \alpha \rceil$, at the end of which the load factor is at $1 - e^{-\lceil \ln \epsilon^{-1} + \alpha \rceil} \ge 1 - \epsilon$. During a given phase $q$, we will use only the first $\dmax$ hashes of each element.

At any given moment, call an element $x$ a \defn{core} element if $\choice(x) \in (\dmax - \dcore, \dmax]$. Our policy for inserting/evicting an element $x$ is:
\begin{itemize}
    \item \textbf{Type 1 Move: }If $x$ is a core element, place it in $h_{\dmax - k + 1}(x)$ for a random $k \in \{1, \ldots, \dcore\}$. If there is an element in that position, it gets evicted.
    \item \textbf{Type 2 Move: }If $x$ is not a core element, then sequentially examine positions $h_i(x)$ for $\max(1, \choice(x) - \dcore) \le i \le \dmax - \dcore$. If we find a free slot, place $x$ there. Otherwise, $x$ becomes a core element, and we use the procedure for placing a core element.
    \item \textbf{Failures: }Finally, if there are $\log^{\omega(1)} n$ Type 1 moves in a row, without the insertion completing, then declare \textbf{failure}.
\end{itemize}

Note that, by construction, there are no core elements at the beginning of a given phase. The only way for an element to become core is if we have at some point first-time probed all of $h_1(x), \ldots, h_{\dmax - \dcore}(x)$. 

As a convention, we will refer to the core elements as forming a \defn{core hash table}. The core hash table is itself a $\dcore$-ary hash table that implements insertions (i.e., additions of new core elements) using random-walk evictions. The failure condition above can be restated as: the overall hash table fails if the core hash table ever has an eviction chain of length $\log^{\omega(1)} n$ (i.e., $\operatorname{polylog} n$ for a sufficiently large polylog).

\paragraph{Selecting Parameters.}
In selecting our parameters $\dcore$ and $\ecore$, we will make black-box use of Theorem \ref{thm:randomwalk}. Let $D : \mathbb{N} \rightarrow \mathbb{R}^+$, $T: \mathbb{N} \rightarrow \mathbb{R}^+$, and $N:\mathbb{N}\rightarrow \mathbb{R}^+$ be the functions from the theorem statement.

We set $\dcore$ and $\ecore = e^{-\dcore + D(\dcore)}$ so that, for all $j \ge \dcore$, we have
\begin{equation} j / e^{j - \alpha} \le \alpha/8 ,\label{eq:alphabeta1} \end{equation}
so that
\begin{equation} 
\ln \ecore^{-1} + \alpha/8 > \dcore \cdot (1 + \ecore),
\label{eq:alphabeta2}
\end{equation}
and so that 
\begin{equation}
    \ecore < 1/2.
    \label{eq:ecoresmall}
\end{equation}

Note that \eqref{eq:alphabeta1} can be achieved simply by setting $\dcore$ to be sufficiently large as a function of $\alpha$. The fact that \eqref{eq:alphabeta2} can be achieved follows from the fact that $D(\dcore) \rightarrow 0$ as $\dcore \rightarrow \infty$. Indeed, this means that $\ln \ecore^{-1}  - \dcore \rightarrow 0$ and that $\ecore \dcore \rightarrow 0$ as $\dcore \rightarrow \infty$, which together ensure that \eqref{eq:alphabeta2} is possible for $\dcore$ sufficiently large. Finally, \eqref{eq:ecoresmall} holds for any large enough $\dcore$, since $\ecore = e^{-\dcore + o(1)}$ as $\dcore \rightarrow \infty$.

Although $\dcore$ is sufficiently large as a function of $\alpha$, it is still $O(1)$. Therefore, a $\dcore$-ary cuckoo hash table, using random-walk evictions, can support insertions at load factors up to $1 - \ecore$ while supporting $O(1)$ expected-time insertions and a failure probability (cumulative across all insertions) of $n^{-\Omega(1)}$. 

\paragraph{Analysis outline. }The proof of Theorem \ref{thm:main} is split into three parts. We begin in Subsection \ref{sec:core} by proving a bound of $1 - \ecore$ on the load factor of the core hash table -- this is the most technical part of the analysis. We then complete the analysis of insertions in Subsection \ref{sec:insertions} and of queries in Subsection \ref{sec:queries}.

\subsection{Bounding the Number of Core Elements}\label{sec:core}

In this subsection we prove that:

\begin{proposition}
With high probability in $n$, the number of core elements at the end of any given phase is at most $(1 - \ecore)n$. 
\label{prop:core2}
\end{proposition}

Throughout the section, we will focus on a fixed phase $q$, and use $\dmax$ to denote the value of $\dmax$ during phase $q$. We break the analysis into two cases, $q = 1$ and $q > 1$, each of which are handled in their own subsection. 

\begin{remark}[A Trick for Handling Failure Events in the Analysis]
    Recall that the core hash table can sometimes fail, in particular, when an insertion is either impossible or takes $\omega(\log n)$ time. When this happens, the overall hash table also fails, and the phase is said to have ended.
    
    So that we can ignore such failure events in this subsection, it is convenient to define the following revised version of the data structure. Suppose that, whenever an insertion of an element $x$ in the core hash table fails (having tried to complete for time $\omega(\log n)$), we complete that insertion as follows: from that point forward in the kickout chain, whenever an element $u$ in the core table is evicted to some other core hash $h_i(u)$, we resample that hash from scratch (even if it has already been probed in the past) and count it as a first-time probe. This modification guarantees that every insertion will eventually succeed.

    We should emphasize that this ``revised data structure'' exists for the sake of analysis only, that is, as a analytical tool for simplifying the proof of Proposition \ref{prop:core2}. The point is that, if we prove Proposition \ref{prop:core2} for the revised data structure, then we have implicitly also proven it for the unrevised one. For the rest of the subsection, we will implicitly consider the revised hash table rather than the un-revised one.
    \label{rem:failures}
\end{remark}

\subsubsection{Analyzing phase $q = 1$}

\begin{lemma}
    During phase $q = 1$, the total number of elements that make it into the core table is, with high probability in $n$, at most $(1 - \ecore)n$.
\end{lemma}
    \begin{proof}
We calculate the total number $P$ of first-time probes made during phase $1$ in two different ways. At the end of the phase, the load factor is $1 - 1 / e^{\dcore + \gamma - \alpha}$. Thus, by Proposition \ref{prop:coupons}, we have with high probability in $n$ that
\begin{equation}P = (\dcore + \gamma - \alpha) n \pm o(n).
\label{eq:P11}
\end{equation}
On the other hand, we also know that
$$P = P_1 + P_2,$$
where $P_1$ is the number of first-time probes made in the core table and $P_2$ is the number of non-core first-time probes. If $Q$ elements make it into the core table, then we have by Proposition \ref{prop:coupons} that, with high probability in $n$, either $Q \le (1 - \ecore)n$, or 
$$P_1 \ge n \ln \ecore^{-1} - o(n).$$
On the other hand, each element that makes it into the core table must first incur $\gamma$ non-core first-time probes. Therefore, either $Q \le (1 - \ecore)n$, or 
$$P_2 \ge \gamma (1 - \ecore) n.$$
Combining these facts, it follows that, with high probability in $n$, either $Q \le (1 - \ecore)n$, or 
\begin{equation}
P \ge n \ln \ecore^{-1} + \gamma \cdot (1 - \ecore) n - o(n).
\label{eq:P21}
\end{equation}
By the construction of $\ecore$, and specifically \eqref{eq:alphabeta2},  we know that 
$$\ln \ecore^{-1} + \alpha/8 > \dcore \cdot (1 + \ecore),$$
which, combined with the fact that $\gamma \le \dcore$, implies that
$$\ln \ecore^{-1} + \gamma \cdot (1 - \ecore) > \dcore + \gamma - \alpha/8.$$ Therefore, \eqref{eq:P11} and \eqref{eq:P21} are contradictory, meaning that 
\eqref{eq:P21} happens with probability at most $1 / \poly(n)$. Thus, we have with high probability in $n$ that $Q \le (1 - \ecore)n$, as desired.
\end{proof}

\subsubsection{Analyzing a phase $q \ge 2$}

For each element $x$, define $\overline{P}(x)$ to be the total number of first-time probes that are made on $x$ in the first $q - 1$ phases, and define $P(x) = \dmax - \dcore - \overline{P}(x).$ Intuitively, $P(x)$ is the number of additional first-time probes that $x$ can make, during phase $q$, before it becomes a core element. A major step in bounding the number of core elements during phase $q$ will be to show that there are many first-time probes performed outside of the core hash table (i.e., first-time probes to the first $\dmax - \dcore$ hashes). For this, an important technical step is to argue that, on average across the elements $x$ inserted/evicted during the phase, $P(x)$ is non-trivially large:

\begin{lemma}
    Suppose $q \ge 2$. Let $E$ be the set of elements that are inserted/evicted during phase $q$. With high probability in $n$, we have
    $$\sum_{x \in E} P(x) \ge |E| \cdot \alpha / 2 - o(n).$$
    \label{lem:Psum}
\end{lemma}
\begin{proof}
Let us break $E$ into $E_1$, consisting of elements inserted during the phase, and $E_2$, consisting of elements that are not inserted during the phase but that are evicted at least once. We know that $\sum_{x \in E_1} P(x) = |E_1| \cdot (\dmax - \dcore) > \alpha |E_1|$, so it suffices to show that
$$\sum_{x \in E_2} P(x) \ge |E_2| \cdot \alpha / 2 - o(n).$$

Let $V$ be the set of elements present at the beginning of the phase. Let us consider how $E_2$ evolves over time, during the phase. A critical observation is that, whenever $|E_2|$ increases, the new element $x$ that is added to $E_2$ is \emph{uniformly random} out of the elements in $V$ that are not yet in $E_2$ -- this is because the element $y$ that is evicting $x$ found $x$ via a first-time probe, and first-time probes are uniformly random. Thus we can think of $E_2$ as being generated by: starting with an empty set, and repeatedly adding random elements from $V$ until some stopping condition is met.

Equivalently, for the sake of analysis, we can think of $E_2$ as being generated by the following fictitious process. Let $v_1, v_2, \ldots$ be independent uniformly random samples of $V$. Then $E_2$ is generated by adding $v_1, v_2, \ldots$ to $E_2$ (where some additions are no-ops since the element has already been added in the past) until some stopping condition is met. 

Without loss of generality, the sequence $v_1, v_2, \ldots$ has total length $O(n \log n)$, since with high probability in $n$, such a sequence will hit every element in $V$. Therefore, if we define $V_t = \{v_1, v_2, \ldots, v_t\}$, we have
$$|E_2| \cdot \alpha/2 - \sum_{x \in E_2} P(x) \le \max_{t \in O(n \log n)} \left(|V_t| \cdot \alpha/2 - \sum_{x \in V_t} P(x) \right).$$
To bound this quantity with high probability, it suffices to consider a fixed $t \in O(n \log n)$ and to prove that, with high probability in $n$, we have
$$|V_t| \cdot \alpha/2 - \sum_{x \in V_t} P(x) \le o(n).$$

Say that the hash table has a \defn{natural} state at the beginning of phase $q$ if, conditioned on that state, the quantity $|V_t| \cdot \alpha/2 - \sum_{x \in V_t} P(x)$ has expected value at most $o(n)$. 

If we condition on the hash table being in a natural state at the beginning of phase $q$, then we can analyze $|V_t| \cdot \alpha/2 - \sum_{x \in V_t} P(x)$ as follows. Because $v_1, \ldots, v_t$ are independent random variables, and because changing any given $v_i$ can change $|V_t| \cdot \alpha/2 - \sum_{x \in V_t} P(x)$ by at most $\dmax - \dcore = O(\log n)$, we can apply McDiarmid's inequality (Theorem \ref{thm:mcdiarmid}) to conclude that, with high probability in $n$, $|V_t| \cdot \alpha/2 - \sum_{x \in V_t} P(x)$ deviates from its mean by at most $\tilde{O}(\sqrt{n})$. This, in turn, implies that $|V_t| \cdot \alpha/2 - \sum_{x \in V_t} P(x) \le o(n) + \tilde{O}(\sqrt{n}) = o(n)$.

Thus, it remains only to show that, with high probability in $n$, the hash table is, in fact, in a natural state at the beginning of phase $q$. Note that the hash table is in a natural state if and only if, for a random element $v \in V$, 
$$\E[P(v)] \ge \alpha / 2 - o(1).$$
This expands to
$$\E[\overline{P}(v)] \le \dmax - \dcore - \alpha/2 + o(1),$$
which is equivalent to saying that the total number of first-time probes made in the first $q - 1$ phases is less than or equal to 
$$|V| \cdot (\dmax - \dcore - \alpha/2 + o(1)).$$
By construction, the load factor at the beginning of phase $q - 1$ is $1 - 1 / e^{\dmax - \dcore - \alpha}$. This implies by Proposition \ref{prop:coupons} that the total number of first-time probes in the first $q - 1$ phases is (with high probability) at most
$$n \cdot (\dmax - \dcore - \alpha + o(1)).$$
Therefore, to complete the lemma, it suffices to show that
$$|V| \cdot (\dmax - \dcore - \alpha/2) \ge n \cdot (\dmax - \dcore - \alpha).$$
Since $|V| = (1 - 1 / e^{\dmax - \dcore - \alpha})n$, this reduces to showing that 
$$n / e^{\dmax - \dcore - \alpha} \cdot (\dmax - \dcore - \alpha/2) \le n \alpha /2.$$
This, in turn, follows from the definition of $\dcore$, which by \eqref{eq:alphabeta1}, satisfies
$$j / e^{j - \alpha} \le \alpha/8$$
for all $j \ge \dcore$.
\end{proof}

Building on Lemma \ref{lem:Psum}, we can prove a lower bound on the total number $P$ of \emph{non-core} first-time probes made during phase $q$. The reason that we care about this lower bound is that we will be able to use it (indirectly) to obtain an \emph{upper bound} on the total number of \emph{core} first-time probes made during the same phase.
\begin{lemma}
    Suppose $q \ge 2$. Let $Q$ be the number of elements during phase $q$ that move into the core hash table. With high probability in $n$, either $Q < n / 2$, or the number $P$ of non-core first-time probes made during phase $q$ satisfies
   $$P \ge \alpha n/8 - o(n).$$
   \label{lem:realfirst-time probes}
\end{lemma}
\begin{proof}
Let $E$ be the set of elements that are inserted/evicted during phase $q$. Let $E_1$ be the subset of $E$ that end up in the core hash table, and $E_2$ be the subset of $E$ that do not. Then, 
$$P \ge \sum_{x \in E_1} P(x) \ge \sum_{x \in E} P(x) - |E_2| \cdot (\dmax - \dcore) \ge \sum_{x \in E} P(x) - (\dmax - \dcore) n / e^{\dmax - \dcore - \alpha},$$
where the final step uses the fact that the number of insertions in phase $q$ is less than $n / e^{\dmax - \dcore - \alpha}$, and that each eviction chain adds at most one element to $E_2$. By the construction of $\dcore$, and specifically by \eqref{eq:alphabeta1}, we know that for all $j \ge \dcore$, we have $j / e^{j - \alpha} \le \alpha/8$. Therefore, 
$$P \ge \sum_{x \in E} P(x) - \alpha n / 8.$$
Finally, applying Lemma \ref{lem:Psum} gives that, with high probability in $n$, either $Q < n/2$ or
$$P \ge \alpha Q / 2 - \alpha n / 8 - o(n) \ge \alpha n / 4 - \alpha n / 8 - o(n) \ge \alpha n / 8 - o(n).$$
\end{proof}

Finally, we can bound the number of elements that make it into the core table during phase $q$.
\begin{lemma}
During any phase $q \ge 2$, the total number of elements that make it into the core table is, with high probability in $n$, at most $(1 - \ecore)n$.
\end{lemma}
\begin{proof}
We calculate the total number $P$ of first-time probes made during phase $q$ in two different ways. First, by Proposition \ref{prop:coupons}, we know that
\begin{equation}P = \dcore n \pm o(n),
\label{eq:P1}
\end{equation}
with high probability in $n$.
On the other hand, we also know that
$$P = P_1 + P_2,$$
where $P_1$ is the number of first-time probes made in the core table and $P_2$ is the number of non-core first-time probes. 
If $Q$ elements make it into the core table, then we have by Proposition \ref{prop:coupons} that, with high probability in $n$, either $Q \le (1 - \ecore)n$, or 
$$P_1 \ge n \ln \ecore^{-1} - o(n).$$
We also have by Lemma \ref{lem:realfirst-time probes} that, with high probability in $n$, either $Q \le (1 - \ecore)n$, or, by \eqref{eq:ecoresmall}, we have $Q > n/2$, and thus that 
$$P_2 \ge n \alpha / 8 - o(n).$$
Combining these facts, it follows that, with high probability in $n$, either $Q \le (1 - \ecore)n$, or 
\begin{equation}
P \ge n \ln \ecore^{-1} + n \alpha / 8 - o(n).
\label{eq:P2}
\end{equation}
By the construction of $\ecore$, and specifically by \eqref{eq:alphabeta2}, we know that $\ln \ecore^{-1} + \alpha / 8 > \dcore + \Omega(1)$. Therefore, \eqref{eq:P1} and \eqref{eq:P2} are contradictory, meaning that \eqref{eq:P2} happens with probability at most $1 / \poly(n)$. Thus, we have with high probability in $n$ that $Q \le (1 - \ecore)n$, as desired.
\end{proof}

\subsection{Bounding insertion time and failure probability}\label{sec:insertions}

We can now bound the insertion time and failure probability for the advanced bubble-up algorithm. This part of the analysis follows a very similar path to the one used for the basic bubble-up algorithm, except that now we use Proposition \ref{prop:core2} in place of Lemma \ref{lem:corenumber} and Theorem \ref{thm:randomwalk} in place of Theorem \ref{thm:2ary}.

We already know from Proposition \ref{prop:core2} that, with high probability in $n$, the number of core elements at any given moment is at most $(1 - \ecore)n$. In addition to this, we will need what we call the \emph{core independence property}: 
\begin{lemma}[The Core Independence Property]
For each element $x$, whether $x$ gets placed in the core hash table during a given phase is independent of the hashes $h_{\dmax - \dcore + 1}(x), \ldots, h_{\dmax}(x)$. 
\label{lem:coreind}
\end{lemma}
\begin{proof} 
This property follows from two observations: (1) Once an element becomes a core element, it stays a core element for the rest of the phase; and (2) Prior to an element $x$ becoming a core element during a given phase, we never evaluate any of $h_{\dmax - \dcore + 1}(x), \ldots, h_{\dmax}(x)$. 
\end{proof}

Lemma \ref{lem:coreind} tells us that we can think of the elements in the core hash table as having $\dcore$ fully random hashes. This lets us think of the core hash table as a standard $\dcore$-ary cuckoo hash table that treats non-core elements as free slots. Whenever a new element is inserted into it (via a Type 2 move), a random eviction chain is performed (via Type 1 moves) until a free slot (i.e., a slot that is either genuinely free or contains a non-core element) is found. Since the load factor of the core hash table is at most $1 - \ecore$ (with high probability), since $\ecore = e^{-\dcore + D(\dcore)}$, and since $\dcore = O(1)$, we can apply Theorem \ref{thm:randomwalk} in order to conclude that:
\begin{itemize}
    \item \textbf{Fact 1: } Each eviction chain in the core hash table has expected length $O(T(\dcore)) = O(1)$;
    \item \textbf{Fact 2: } The probability that any eviction chain in the core hash table ever has length $T(\dcore) \log^{\omega(1)} n = \log^{\omega(1)} n$ (this includes the event that an insertion fails) is at most $n^{-\Omega(1)}$.
\end{itemize}

Recall that the advanced bubble-up algorithm fails if there are ever $\log^{\omega(1)} n$ Type 1 moves in a row. Fact 2 tells us that the probability of such a failure ever occurring during a given phase is $n^{-\Omega(1)}$. Since the number of phases is $d = O(\log n)$, it follows that the probability of any failures ever occurring is $n^{-\Omega(1)}$. 

Fact 1, on the other hand, can be used to bound the expected insertion time, overall, within the (full) hash table. Consider the $((1 - \delta)n + 1)$-th insertion. Let $Q$ be the number of first-time probes made by the insertion. Since each first-time probe has at least a $\delta$ probability of finding a free slot, we have that
$$\E[Q] = O(\delta^{-1}).$$
On the other hand, we can bound the total time $T$ spent on the insertion by the sum of two terms:
\begin{itemize}
    \item $T_1$ is the number of first-time probes made by Type 1 and moves;
    \item $T_2$ is the time spent on Type 2 moves.
\end{itemize}

By construction, $T_2 \le O(Q)$, so $\E[T_2] = O(\delta^{-1})$. To bound $T_1$, define $J$ to be the total number of core-table eviction chains that occur during the current insertion. By Fact 1, we have that
$$\E[T_1] \le O(J).$$
On the other hand, each core-table eviction chain is triggered by an element $x$ becoming a core element, at which point at least one of $h_{\dmax - \dcore + 1}(x), \ldots, h_{\dmax}(x)$ experiences a first-time probe. It follows that $J \le Q$, which implies that $\E[T_1] \le O(\delta^{-1})$. Thus $\E[T_1 + T_2] = O(\delta^{-1})$, as desired.

\subsection{Bounding positive query time}\label{sec:queries}


In this section, we prove the following proposition:
\begin{proposition}[Bounding positive query time by $O(1)$]
Consider an element $x$ that is in the hash table. Then the time to query it by examining $h_{\dmax}(x), h_{\dmax - 1}(x), \ldots$ is bounded above by a geometric random variable with mean $O(1)$.
\label{prop:query}
\end{proposition}

The basic idea for bounding the query time is as follows. We will show that, within each phase, each element $x$ has probability at least $\Omega(1)$ of being evicted at least once. If $x$ is evicted, then we will argue that it has a good probability of being placed into the core hash table (for the current phase). This means that an $\Omega(1)$ fraction of elements are in the core hash table for the current phase; an $\Omega(1)$ fraction of the \emph{remaining} elements were in the core hash table for the previous phase; and so on. If an element $x$ was in the core hash table $j$ phases ago, then it is currently in a position of the form $h_{\dmax - O(j)}(x)$. This means that the element can be queried in time $O(j)$, which we will argue is a geometric random variable with mean $O(1)$. Although this is the basic structure of the proof, the full proof is complicated by two considerations: (1) When an element $x$ is evicted, it does not necessarily get placed into the core hash table; and (2) The hash table can, with some small probability, fail. 

We begin by arguing that, whenever an element is inserted/deleted, it is very likely placed into the core table.
\begin{lemma}
Condition on an arbitrary state for the hash table in phase $i$, and condition on some element $x$ being inserted/evicted at least once during phase $i$. Then, with probability $1 - O(e^{-i})$, $x$ is placed into the core table for phase $i$.
\label{lem:coreplace}
\end{lemma}
\begin{proof}
In order for $x$ to not be in the core table, one of the not-yet-probed slots from the sequence $h_1(x), h_2(x), \ldots, h_{\dmax - \dcore}(x)$ must be free at the beginning of the phase. Since the free-slot density during the phase is $O(e^{-i \dmax})$, the probability of this occurring is at most
$$(\dmax - \dcore) \cdot O(e^{-i\dmax}) = O(\dmax e^{-i \dmax}) = O(e^{-i}).$$
\end{proof}

Next we bound the probability of certain rare events occurring. Let $A_i$ be the indicator random variable that either phase $i$ fails, or that during phase $i$, fewer than $n$ first-time probes of the form $h_i(y)$, $i \in (\dmax - \dcore, \dmax]$, are made. 
\begin{lemma}
We have $\Pr[A_i = 1] \le O(1/n)$.
\label{lem:noweirdstuff}
\end{lemma}
\begin{proof}
The probability of phase $i$ failing is $O(1/n)$. Supposing the phase does not fail, by Proposition \ref{prop:coupons}, we have with probability $1 - 1 / \poly(n)$ that the total number of first-time probes made by the end of phase $i$ is at least $n (\dmax - \alpha) - o(n)$. The number of probes made outside of the core table is at most $n (\dmax - \dcore)$, so at least $n (\dcore - \alpha) - o(n) > n$ probes are made within the core table, as desired.
\end{proof}

For an element $x$ and a phase $i$, let $B_{i, x}$ be the indicator random variable event $$(A_i = 1 \text{ or } x \text{ is inserted/evicted during phase }i),$$ and let $C_{i, x}$ be the indicator random variable event that $$(A_i = 1 \text{ or } x \text{ gets placed in core hash table during phase } i).$$ Finally, let $B_{< i, x} := \{B_{j, x}\}_{j = 1}^{i - 1}$ and $C_{< i, x} := \{C_{j, x}\}_{j = 1}^{i - 1}$.

We will argue that $B_{i, x}$ has probability at least $\Omega(1)$ of being $1$, even if we condition on information about earlier phases; and that, if we condition on $B_{i, x}$ being $1$, then $C_{i, x}$ is very likely to also be $1$.
\begin{lemma}
For an element $x$ that is inserted at some point in the first $i$ phases, and for any outcomes of $B_{< i, x}$ and $C_{< i, x}$, we have
\begin{equation*}\Pr[B_{i, x} = 1 \mid B_{< i, x}, C_{< i, x}] \ge 1 - 1 /e
\end{equation*}
and 
$$\Pr[C_{i, x} = 1 \mid B_{i, x} = 1, B_{< i, x}, C_{< i, x}] \ge 1 - O(e^{-i}).$$
\label{lem:twoid}
\end{lemma}
\begin{proof}
If $x$ is inserted during phase $i$, then $B_{i, x}$ holds trivially. Suppose $x$ is inserted in a previous phase. By definition, either $A_i = 1$ or at least $n$ first-time probes will be made in phase $i$. For the sake of analysis, if $A_i = 1$, and if the total number of first-time probes that we make is $n' < n$, let us imagine that we make $n - n'$ additional (artificial) first-time probes, so that the total number is $n$. This means that, regardless of whether $A_i = 1$, we make at least $n$ first-time probes; and that, so long as at least one of these probes finds the position $j$ containing $x$ at the beginning of the phase, then we will have $B_{i, x} = 1$ (in particular, if the probe that finds $x$ is artificial, then this implies $A_i = 1$, which implies $B_{i, x} = 1$ trivially). Thus, to lower bound $\Pr[B_{i, x} = 1 \mid B_{< i, x}, C_{< i, x}]$, it suffices to lower bound the probability that, if we perform $n$ first-time probes, we find the position containing $x$. This means that
$$\Pr[B_{i, x} = 1 \mid B_{< i, x}, C_{< i, x}] \ge 1 - (1 - 1/n)^{n} \ge 1 - 1 / e,$$
as desired.

Finally, suppose that $B_{i, x}$ occurs, and condition on any outcomes for $ B_{< i, x}, C_{< i, x}$. If $B_{i, x}$ occurs without $x$ getting evicted inserted/evicted during phase $i$, then it must be that $A_i$ occurs, which implies that $C_{i, x}$ occurs. On the other hand, if $x$ gets evicted/inserted during phase $i$, then the only way for $C_{i, x}$ to be $0$ is if, when $x$ gets evicted/inserted, it does not get placed into the core hash table. We know from Lemma \ref{lem:coreplace} that the probability of this occurring is at most $O(e^{-i}),$
as desired.
\end{proof}

We now prove that, at the end of any given phase, the query time of each element is bounded above by a geometric random variable. Once we prove this, the only remaining task will be to consider intermediate time points between the ends of phases.
\begin{lemma}
    Consider an element $x$ that is inserted at some point in the first $i$ phases. At the end of phase $i$, let $j$ be the minimum $j$ such that $x$ is in position $h_{\dmax - j}(x)$ (and set $j = 0$ if the hash table fails before the end of phase $i$). Then, $j$ is bounded above by a geometric random variable with mean $O(1)$.
    \label{lem:mainquery}
\end{lemma}
\begin{proof}
Note that $j \le \dmax \le O(i)$ trivially. By Lemma \ref{lem:noweirdstuff}, we have $\Pr[A_i = 1] = O(1/n) = e^{-\Omega(i)}$. Therefore, $j \cdot A_i$ is bounded above by a geometric random variable with mean $O(1)$. To complete the proof, we will argue that $j \cdot (1 - A_i)$ is also bounded above by a geometric random variable with mean $O(1)$. 

Let $i_0$ be the largest $i_0 \le i$ such that $x$ was either inserted or evicted at some point during phase $i_0$. (Note that $i_0$ itself is a random variable). Observe that, if $C_{i_0, x} = 1$, then either $A_i$ holds, or $x$ is placed into the core hash table during phase $i_0$. In the latter case, we have $j \le \dcore \cdot (i - i_0 + 1) = O(i - i_0 + 1)$. Thus,
\begin{equation} j \cdot C_{i_0, x} \cdot (1 - A_i) = O(i - i_0 + 1) \cdot (1 - A_i).
\label{eq:jcase1}
\end{equation}
On the other hand, if $C_{i_0, x} = 0$, then $j \le \dmax = O(i)$ trivially, so
\begin{equation} j \cdot (1 - C_{i_0, x}) \cdot (1 - A_i) = O(i).
\label{eq:jcase2}
\end{equation}

To complete the proof, we will show that each of the left-hand sides of \eqref{eq:jcase1} and \eqref{eq:jcase2} are bounded above by geometric random variables with means $O(1)$. This will imply that $j \cdot (1 - A_i)$ is also bounded above by a geometric random variable with mean $O(1)$, as desired.

For \eqref{eq:jcase1}, it suffices to show that $(i - i_0) \cdot (1 - A_i)$ is bounded above by a geometric random variable with mean $O(1)$. By the definition of $i_0$, if $A_i = 0$, then all of $B_{i_0 + 1, x}, B_{i_0 + 2, x}, \ldots, B_{i, x}$ are $0$. By Lemma \ref{lem:twoid}, the probability of this occurring for a given $i_0$ is at most $(1 - 1 /e)^{i - i_0}$. It follows that, for all $k \ge 0$, we have $\Pr[(i - i_0) \cdot (1 - A_i) \ge k] \le (1 - 1/e)^k$. This means that $(i - i_0) \cdot (1 - A_i)$ is bounded above by a geometric random variable with mean $O(1)$, as desired.

For \eqref{eq:jcase2}, it suffices to show that $\Pr[C_{i_0, x} = 0 \text{ and }A_i = 0] = e^{-\Omega(i)}$. This would imply that $j \cdot (1 - C_{i_0, x}) \cdot (1 - A_i)$ is non-zero with probability $e^{-\Omega(i)}$. Since, even when $j \cdot (1 - C_{i_0, x}) \cdot (1 - A_i)$ is non-zero, it is at most $O(i)$, it would follow that $j \cdot (1 - C_{i_0, x}) \cdot (1 - A_i)$ is bounded above by a geometric random variable with mean $O(1)$. 

By the definition of $i_0$, we have that $B_{i_0, x} = 1$ and that, if $A_i = 0$, then $B_{i_0 + 1, x}, \ldots, B_{i, x} = 0$. Thus $\Pr[C_{i_0, x} = 0 \text{ and } A_i = 0]$ is at most the probability that there exists any value $k$ for $i_0$ such that: (1) $C_{k, x} = 0$, (2) $B_{k, x} = 1$, and (3) $B_{k + 1, x}, \ldots, B_{i, x} = 0$. By Lemma \ref{lem:twoid}, we have for any given value of $k$ that this occurs with probability at most
$$\Pr[C_{k, x} = 0 \mid B_{k, x} = 1, B_{<k, x}, C_{<k, x}] \cdot \prod_{r = k + 1}^{i} \Pr[B_{r, x} = 0 \mid B_{<r, x}, C_{<r, x}] \le e^{-\Omega(k)} \cdot \prod_{r = k + 1}^{i} 1/e \le e^{-\Omega(i)}.$$
Applying a union bound over the $O(i)$ options for $k$, the probability of any $k$ existing that satisfies all three of the above conditions is $O(i e^{-\Omega(i)}) = e^{-\Omega(i)}$. This implies that $\Pr[C_{i_0, x} = 0 \text{ and } A_i = 0] = e^{-\Omega(i)}$, as desired.
\end{proof}

Finally, with a bit of additional casework to handle elements that are inserted during the current (unfinished) phase, we can prove Proposition \ref{prop:query}.
\begin{proof}[Proof of Proposition \ref{prop:query}]
If $x$ was inserted prior to the current phase $i$, then the proposition follows from Lemma \ref{lem:mainquery} applied at the end of phase $i - 1$. If $x$ was inserted during phase $i$, and the hash table does not fail prior to the end of $x$'s insertion, then we have by Lemma \ref{lem:twoid} that, with probability $1 - O(e^{-i})$, $x$ is placed in the core hash table and therefore that $j = O(1)$. In the $O(e^{-i})$-probability case that $x$ is not placed into the core hash table, we have trivially that $j \le \dmax \le O(i)$. It follows that $j$ is bounded above by a geometric random variable with mean $O(1)$, as desired.
\end{proof}

\section{Conclusion}
We have introduced \emph{bubble-up cuckoo hashing}, a variation of $d$-ary cuckoo hashing that achieves all of the following properties:
\begin{itemize}
    \item uses $d = \lceil \ln \epsilon^{-1} + \alpha \rceil$ hash locations per item for an arbitrarily small positive constant $\alpha$.
    \item achieves expected insertion time $O(\delta^{-1})$ for any insertion taking place at load factor $1 - \delta \le 1 - \epsilon$.
    \item achieves expected positive query time $O(1)$, independent of $d$ and $\epsilon$.
\end{itemize}

\noindent Several major open questions remain. 
\begin{enumerate}
    \item Do simpler algorithms (e.g., random-walk or BFS) already get good time bounds (e.g., $\poly \epsilon^{-1}$) when $d = \Theta(\ln \epsilon^{-1})$? It is widely believed that the answer should be yes, but proving this remains difficult. 
    \item Can one hope for $O(\epsilon^{-1})$-time operations even when $d = \ln \epsilon^{-1} + o(1)$? Our bounds require $d \ge \ln \epsilon^{-1} + \alpha$ for some small but positive constant $\alpha$.
    \item Can our results for bubble-up cuckoo hashing be extended to work with explicit families of hash functions, for example, with variants of tabulation hashing \cite{tabulation1, tabulation2, tabulation3, tabulation4}?
    \item Do the techniques from bubble-up cuckoo hashing come with lessons for real-world hash tables?
\end{enumerate}

In addition to these, another major open question is to develop efficient insertion algorithms for \emph{bucketized cuckoo hashing}, where each element hashes to two buckets of size $b$ \cite{dietzfelbinger2007balanced}. This setting appears to be harder to analyze than the $d$-ary case because, as we increase $b$, the total amount of \emph{randomness} that we have to work with does not increase. Because each item hashes to only two buckets, any item that is evicted more than \emph{once} in its lifetime will be forced to make use of spoiled randomness (randomness that has already affected the hash-table state in the past). This issue of spoiled randomness seems to be a major challenge for the analysis of the bucketized version of the data structure.

Within the study of bucketized cuckoo hashing, there are several goals that would be interesting to accomplish. (See, also, the discussion of bucketized cuckoo hashing in the related-work portion of the introduction.) Major questions include:
\begin{enumerate}
    \item Can one achieve $\poly \epsilon^{-1}$-time insertions with buckets of size $b = \Theta(\log \epsilon^{-1})$? 
    \item If the bucket size $b$ is a constant, can one get arbitrarily close to the critical load threshold (the maximum load at which a valid hash-table configuration exists) while still supporting $O(1)$-expected time insertions? 
    \item What can one say about the random-walk and BFS algorithms, in particular. Do they achieve either of the aforementioned goals? 
\end{enumerate}


\section*{Acknowledgments}
We thank Tolson Bell and Alan Frieze for discussions explaining the nuances their paper \cite{bell20241}.  In particular, they explained that one can extend their main theorem to support adversarial removal of some elements, and that the way to do this is not through any direct monotonicity argument, but rather by carefully retracing the sequence of arguments leading to their main theorem.  We use this in Appendix~\ref{sec:impchoice}.

Michael Mitzenmacher was supported in part by NSF grants CCF-2101140, CNS-2107078, and DMS-2023528. William Kuszmaul was partially supported by a Harvard Rabin Postdoctoral Fellowship and by a Harvard FODSI fellowship under NSF grant DMS-2023528.

\bibliographystyle{plain}
\bibliography{name}

\appendix

\section{Proof of Proposition \ref{prop:coupons}}\label{sec:appendix}

To prove Proposition \ref{prop:coupons}, we begin with the following claim:

\begin{claim}
Let $k \le O(n \log n)$, and let $X_1, X_2, \ldots, X_k$ be iid uniformly random elements of $[n]$. Also, let $X = \bigcup_i \left  \{ X_i \right \}$. Then, with high probability in $n$,
$$\big| X \big| = n\cdot  (1 - e^{-k/n}) \pm O(\sqrt{n} \log n).$$
\label{clm:coupons}
\end{claim}
\begin{proof}
The probability that a given $j \in [n]$ is not in any of $X_1, \ldots, X_k$ is 
$$(1 - 1/n)^k = e^{-k/n} \pm O(1/n).$$
It follows that
$$\E\left[\big| X \big|\right] = n\cdot  (1 - e^{-k/n}) \pm O(1).$$
Moreover, if we define $f(X_1, X_2, \ldots, X_k) = |X|$, then $f$ is a function of $k$ independent random variables, each of which affects $f$'s value by at most $\pm 1$ (that is, changing $X_i$ for some $i$ changes $f$ by at most $1$). Thus we can apply McDiarmid's inequality (Theorem \ref{thm:mcdiarmid}) to obtain the concentration bound
$$\Pr[|f - \E[f]| \ge t \sqrt{k}] = e^{-\Omega(t^2)}.$$
Plugging in $t = \Theta(\sqrt{\log n})$ allows us to conclude that, with probability $1 - 1 / e^{\Omega(t^2)} = 1 - 1 / \poly(n)$, we have
$$\big| X \big| = n\cdot  (1 - e^{-k/n}) \pm O\left(t \sqrt{k}\right) = n\cdot  \left(1 - e^{-k/n}\right) \pm O\left(\sqrt{n} \log n\right).$$
\end{proof}

We now prove Proposition \ref{prop:coupons}:
\propcoupons*
\begin{proof}[Proof of Proposition \ref{prop:coupons}]
    By Claim \ref{clm:coupons}, there exists a positive constant $c$, such that, with high probability in $n$, we have:
    \begin{itemize}
        \item The first $n \ln (\epsilon^{-1} - c\log n / \sqrt{n})$ coupons sample more than $(1 - \epsilon) n$ distinct values.
        \item The first $n \ln (\epsilon^{-1} + c \log n / \sqrt{n})$ coupons sample fewer than $(1 - \epsilon) n$ distinct values.
    \end{itemize} 
    It follows that the number of coupons needed to sampled $(1 - \epsilon)n$ distinct values is in the range
    $$n \ln (\epsilon^{-1} \pm  O(\log n / \sqrt{n})).$$
    So long as $\epsilon \ge n^{-1/4}$, this range is contained in the range
    \begin{align*}
    n \ln (\epsilon^{-1} \cdot (1 \pm O(\log n / n^{1/4}))) & = n \ln \epsilon^{-1} \pm O(\log n / n^{1/4}). \\
    \end{align*}
    
\end{proof}

\section{Implementing $\choice(x)$}\label{app:choice}
\label{sec:impchoice}

In the body of the paper, we assume access to a constant-time function $\choice(x)$ that, for a given element $x$, identifies which hash function $h_i$ was used to place $x$ in its current position. In this section, we will show how to remove this assumption while preserving the final bounds that we achieve.

Note that we only ever invoke $\choice(x)$ when we already know the current position $j$ containing $x$. So we can assume that $\choice(x)$ actually takes two arguments $x$ and $j$. Our approach in this section will be to replace $\choice(x, j)$ with an explicit protocol $\choice'(x, j)$ that (with a bit of additional algorithmic case-checking) preserves both the correctness and time guarantees from the main sections of the paper. So that we can discuss both the basic bubble-up algorithm and the advanced bubble-up algorithm simultaneously, we will use $\dmax$ to denote $d$ for the basic algorithm and $\dmax$ for the advanced algorithm; we will use $\dcore$ to mean $2$ for the basic algorithm and to mean $\dcore$ for the advanced algorithm; and we will use $\ecore$ to mean $0.51$  for the basic algorithm and to mean $\ecore$ for the advanced algorithm. 

The protocol $\choice'(x, j)$ examines positions $h_{\dmax}(x), h_{\dmax - 1}(x), \ldots$, one after another, and returns
\begin{equation} \max \{i \le \dmax \mid h_i(x) = j\}.
\label{eq:choiceprime}
\end{equation}
Note that, even though $\choice'(x, j)$ evaluates $h_{\dmax}(x), h_{\dmax - 1}(x), \ldots$, these do not count as probes (and, specifically, first-time probes) in the body of the paper.

To argue that we can use $\choice'$ in place of $\choice$, there are two issues we must be careful about:
\begin{itemize}
    \item \textbf{Issue 1: }$\choice(x, j)$ does not necessarily equal $\choice'(x, j)$. In particular, if there exist $i_1 < i_2 \le \dmax$ such that $h_{i_1}(x) = h_{i_2}(x)$, then we could have $\choice(x, j) = i_1$ but $\choice'(x, j) = i_2$. 
    \item \textbf{Issue 2: }$\choice'(x, j)$ is not a constant-time function. Thus we must analyze its time contribution to each insertion.
\end{itemize}

\paragraph{Handling Issue 1.} We begin by handling Issue 1, as it is the more significant of the two. Note that, in our algorithms, if $\choice(x, j) > \dmax - \dcore$, then we do not care about the specific value that $\choice(x, j)$ takes in the range $(\dmax - \dcore, \dmax]$. Thus, the case where Issue 1 can be a problem is if $i_1 \le \dmax - \dcore$. 

Call an element $x$ \defn{corrupt} if there exists $i_1 \le \dmax - \dcore$ and $i_2$ satisfying $i_1 < i_2 \le \dmax$ such that $h_{i_1}(x) = h_{i_2}(x)$. These are the elements that, at any given moment, cause Issue 1 to be a problem. It is worth noting that, with high probability, there are very few such elements.

\begin{lemma}
    At any given moment, we have with probability $1 - 1 / \poly(n)$ that there are at most $O(\log^3 n)$ corrupt elements.
    \label{lem:corrupt}
\end{lemma}
\begin{proof}
Each element independently has at most an $O(\dmax^2 / n)$ probability of being corrupt. Thus, by a Chernoff bound, we have with probability $1 - 1 / \poly(n)$ that the number of corrupt elements is at most $O(\dmax^2 \log n) = O(\log^3 n)$.
\end{proof}

The fact that there are $O(\log^3 n)$ corrupt elements means that they have negligible impact on the various quantities addressed in the analysis. This includes both intermediate quantities used in the analysis (e.g., the total number of non-core-hash-table first-time probes made during some time period) as well as the main quantity that we actually care about (the number of elements in the core hash table). Thus, one can proceed with exactly the same analyses as in the body of the paper in order to achieve the desired bound of $(1 - \ecore)n$ on the number of elements in the core hash table.

The other place where Issue 1 shows up is more subtle. Recall that, in order so that we can treat the core hash table as a cuckoo hash table, we need what we call the \emph{core independence property}: that, for each element $x$, when $x$ gets placed into the core hash table, the hashes $h_{\dmax - \dcore + 1}(x), \ldots, h_{\dmax}(x)$ are still fully random (and independent of the state of the core hash table). 

\emph{A priori}, a corrupt element could cause us to break the core independence property as follows. Suppose that $h_i(x) = h_j(x)$ for some $i \le \dcore - \dmax < j$. When $x$ performs a first-time probe on $h_i(x)$, it will use the position if and only if it is vacant. However, if $x$ does use $h_i(x)$, then in the future, the $\choice'$ function will conclude that $x$ is in the core table. Thus, whether or not the $\choice'$ function in the future perceives $x$ to be in the core table depends on the state of position $h_i(x)$ right now, which depends on the state of the core table right now. 

To resolve this issue, we add one final modification to the algorithm: whenever we perform a probe on $h_i(x)$ for some $i \le \dmax - \dcore$, we also check (in $O(1)$ time) whether $h_i(x) = h_j(x)$ for any $j \in (\dmax - \dcore, \dmax]$; if such a collision occurs, then we immediately place $x$ into the core table (\emph{without needing to examine the contents of position $h_i(x)$ first}). 

This still does not give us the full core-independence property, but it does give us a slightly weaker version of the property: when an element $x$ is placed into the core hash table, the only thing that we have revealed about its core hashes so far is that they are \emph{not} equal to any of the non-core hashes $h_1(x), h_2(x), \ldots, h_{i - 1}(x)$ for $x$ that we have performed first-time probes on so far. Another way to think about this is that, when an element $x$ is placed into the core hash table, there is some known set  $F$ of $O(\log n)$ \emph{forbidden} hashes (i.e., the non-core hashes that we have already verified are distinct from $x$'s core hashes), and that $x$'s core hashes are uniformly and independently random from the set $[n] \setminus F$. Call this the \emph{weak core-independence property}. 

To apply the weak core-independence property, we need slightly \emph{stronger} versions of Theorems \ref{thm:2ary} and \ref{thm:randomwalk}:

\begin{proposition}
    Let $d = O(1)$. Consider the setups in Theorems \ref{thm:2ary} and  \ref{thm:randomwalk}, and suppose that, before each insertion, an adaptive adversary (who gets to see the current hash table state but not the hashes of future elements) is permitted to choose a set of $O(\log n)$ \emph{forbidden hashes} $F \subseteq [n]$ for the next insertion. This means that the next element $x$ to be inserted then has its hashes drawn uniformly and independently random from $[n] \setminus F$. Even with this modification, the conclusions of the theorems continue to hold.
    \label{prop:modifiedthms}
\end{proposition}

Proposition \ref{prop:modifiedthms} tells us that the weak version of the core independence property is still sufficient for us to obtain the desired guarantees from the core hash table. We will defer the proof of the proposition the end of the section. However, assuming that it is true, this completes our solution to Issue 1.

\paragraph{Handling Issue 2.} Issue 2 concerns the time spent evaluating the $\choice'$ function on any given insertion. 

In both the basic and advanced versions of the bubble-up algorithm, we have the following property: if an element $y$ is part of an eviction chain, and if it is not the \emph{final} member of the eviction chain, then the time spent on $y$ is at least $\Omega(\dmax - \choice'(y))$. Thus,  with the exception of the final element $z$ in the eviction chain, the $O(\dmax - \choice'(x))$ time spent evaluating $\choice'$ can be amortized to the time already spent on the insertion. The final element in the eviction chain spends at most $O(\dmax)$ time evaluating $\choice'$. For the basic bubble-up algorithm, this adds at most $O(\dmax) = O(d) = O(\ln \epsilon^{-1})$ time to each insertion, which is already permitted in Theorem \ref{thm:basic}. For the advanced bubble-up algorithm, if the hash table is at load factor $1 - \delta$, then $\dmax = O(\log \delta^{-1}) = O(\delta^{-1})$, so the time spent on the final call to $\choice'$ does not change the asymptotic expected time spent on the insertion.

\paragraph{Proving Proposition \ref{prop:modifiedthms}.} We complete the section by proving Proposition \ref{prop:modifiedthms}. We begin with an easier proposition:

\begin{proposition}
    Consider the setups in Theorems \ref{thm:2ary} and  \ref{thm:randomwalk}, and suppose that, before each insertion, an adaptive adversary (who gets to see the current hash table state but not the hashes of future elements) gets to remove some subset of the elements and rearrange the remaining elements arbitrarily. Even with this modification, the conclusions of the theorems continue to hold.
    \label{prop:modifiedthmsa}
\end{proposition}

Proposition \ref{prop:modifiedthmsa} follows from the the same sequence of arguments as for the original versions of Theorems \ref{thm:2ary} \cite{pagh2004cuckoo} and \ref{thm:randomwalk} \cite{bell20241}. In both cases, the analysis still works even if the initial arrangement of the elements is arbitrary (indeed, this is explicit in \cite{bell20241}), and even if some elements are omitted by an adaptive adversary.

Using Proposition \ref{prop:modifiedthmsa}, we can establish Proposition \ref{prop:modifiedthms} as follows:

\begin{proof}[Proof of Proposition \ref{prop:modifiedthms}]
Suppose that, for the $i$-th insertion, we have forbidden set $F_i$. To generate the $i$-th insertion $x_i$, we will generate a sequence of elements $x^{(1)}_i, x^{(2)}_i, \ldots$, each of which has fully random hashes, and we will set $x_i$ to be the first one that respects the forbidden set $F_i$. If $x_i = x^{(j)}_i$ for some $j > 1$, then we consider $x^{(1)}_i, x^{(2)}_i, \ldots, x^{(j - 1)}_i$ to be \emph{phantom elements}. That is, we will imagine that these elements actually \emph{were} part of the insertion sequence, but that an adversary simply chooses to omit them during all future insertions. (This adversary is in the style of the adversary from Proposition \ref{prop:modifiedthmsa}.)

Note that phantom elements are extremely rare. Each insertion independently has a probability at most $O(|F| / n) = \tilde{O}(1/n)$ of being a phantom element, so by a Chernoff bound, the total number of phantom elements across all insertions is, with very high probability, at most $\operatorname{polylog} n$. These phantom elements therefore have a negligible overall effect on the total number of insertions that we perform.

Applying Proposition \ref{prop:modifiedthmsa}, we can conclude that, even with an adversary omitting the phantom elements, the expected time per insertion remains $O(1)$. This implies Proposition \ref{prop:modifiedthms}, as desired.
\end{proof}

\section{Supporting deletions with tombstones}\label{app:tombstones}

In this section, we describe how to use tombstones to prove the following corollary of Theorem \ref{thm:main}:

\cormain*
\begin{proof}
We will implement deletions by \emph{marking} the appropriate element as deleted. If an element marked as deleted is later reinserted, it is simply un-marked. Elements that are marked as deleted are referred to as \defn{tombstones}. Tombstones participate in eviction chains just like any other element -- from the perspective of insertions, they are regular elements. 

Because insertions treat tombstones as elements, one must be careful to limit the \defn{augmented load factor} of the hash table, which is the load factor \emph{including} the tombstones. Whenever the augmented load factor reaches $1 - \epsilon'$, for some $\epsilon' = \Theta(\epsilon)$ to be chosen later, the hash table is rebuilt from scratch and the tombstones are cleared out. These rebuilds occur every $O(\epsilon n)$ operations, and cost 
$$\int_{m = 0}^{(1 - \epsilon)n - 1} O(n / (n - m))^{-1} \text{d}m = O(n \log \epsilon^{-1})$$ 
expected time. The amortized expected rebuild cost per insertion is therefore $$O\left(\frac{n \log \epsilon^{-1}}{\epsilon n} \right) = O(\epsilon^{-1} \log \epsilon^{-1}).$$ 

The tombstones raise the maximum load factor that the hash table must support from $1 - \epsilon$ to $1 - \epsilon'$. Setting $\epsilon' = e^{\alpha} \epsilon$, we can use Theorem \ref{thm:main} to get $d = \lceil \ln ((\epsilon')^{-1}) + \alpha \rceil = \lceil \ln \epsilon^{-1} + 2\alpha \rceil$. Finally, by using $\alpha/2$ in place of $\alpha$, we obtain the desired overall bound on $d$.

Note that, if an insertion failure occurs, we can handle it by just performing an immediate rebuild. The expected number of times that this happens during a given time window between scheduled rebuilds is $n^{-\Omega(1)}$ (by Theorem \ref{thm:main}), so these extra failure-induced rebuilds contribute negligibly to the overall amortized expected insertion cost. 
\end{proof}
\begin{remark}
    One might naturally wonder whether rebuilds themselves can be implemented space efficiently, i.e., without temporarily increasing the space usage during the rebuild. There are several standard approaches that one can use to do this. One approach is to use Dietzfelbinger and Rink's so-called \emph{splitting trick} \cite{dietzfelbinger2009applications}, in which the hash table is partitioned into, say, $\sqrt{n}$ pieces, each of which is implemented as its own Cuckoo hash table; each element hashes to a random piece; and, during rebuilds, the pieces are rebuilt one at a time. Each of these rebuilds can be performed using a temporary array of size $O(\sqrt{n})$. So long as $\epsilon^{-1} = n^{o(1)}$, it is straightforward to apply this technique in order to reduce the overall space overhead of rebuilds to a $(1 + o(\epsilon))$ factor. For a detailed discussion of this technique in the context of resizing, see also \cite{bender2024modern}.
    \label{rem:rebuilds}
\end{remark}

\end{document}